\documentclass[lettersize,journal]{IEEEtran}
\IEEEoverridecommandlockouts

\usepackage{amssymb}
\usepackage[cmex10]{amsmath}
\usepackage{stfloats}
\usepackage{graphicx}
\usepackage{subfigure}
\usepackage{tabularx}
\usepackage{epsfig,epsf,color,balance,cite}
\usepackage{verbatim}
\usepackage{url}
\usepackage{bm}
\usepackage{booktabs}
\usepackage{siunitx}

\usepackage{amsmath}
\allowdisplaybreaks[2]
\usepackage{ntheorem}
\newtheorem{theorem}{Theorem}

\newtheorem*{proof}{Proof}

\usepackage{algorithm}
\usepackage{algorithmic}

\usepackage{color}
\definecolor{myc1}{rgb}{0,0,0}
\begin{document}

\title{Agentic AI for Low-Altitude Semantic Wireless Networks: An Energy Efficient Design}

\author{Zhouxiang Zhao, 
        Ran Yi,
        Yihan Cang, 
        Boyang Jin,
        Zhaohui Yang, 
        Mingzhe Chen,~\IEEEmembership{Senior Member,~IEEE,}
        Chongwen Huang, 
        and Zhaoyang Zhang,~\IEEEmembership{Senior Member,~IEEE}

\thanks{Zhouxiang Zhao, Ran Yi, Boyang Jin, Zhaohui Yang, Chongwen Huang, and Zhaoyang Zhang are with the College of Information Science and Electronic Engineering, Zhejiang University, and also with Zhejiang Provincial Key Laboratory of Info. Proc., Commun. \& Netw. (IPCAN), Hangzhou 310027, China (e-mails: \{zhouxiangzhao, ranyi, byjin10225, yang\_zhaohui, chongwenhuang, ning\_ming\}@zju.edu.cn).}
\thanks{Yihan Cang is with National Mobile Communications Research Laboratory, Southeast University, Nanjing 211189, China (e-mail: yhcang@seu.edu.cn).}
\thanks{Mingzhe Chen is with Department of Electrical and Computer Engineering and Institute for Data Science and Computing, University of Miami, Coral Gables, FL 33146, USA (e-mail: mingzhe.chen@miami.edu).}
\vspace{-2em}
}

\maketitle

\begin{abstract}
This letter addresses the energy efficiency issue in unmanned aerial vehicle (UAV)-assisted autonomous systems. We propose a framework for an agentic artificial intelligence (AI)-powered low-altitude semantic wireless network, that intelligently orchestrates a sense-communicate-decide-control workflow. A system-wide energy consumption minimization problem is formulated to enhance mission endurance. This problem holistically optimizes key operational variables, including UAV's location, semantic compression ratio, transmit power of the UAV and a mobile base station, and binary decision for AI inference task offloading, under stringent latency and quality-of-service constraints. To tackle the formulated mixed-integer non-convex problem, we develop a low-complexity algorithm which can obtain the globally optimal solution with two-dimensional search. Simulation results validate the effectiveness of our proposed design, demonstrating significant reductions in total energy consumption compared to conventional baseline approaches.
\end{abstract}


\IEEEpeerreviewmaketitle

\section{Introduction}
\IEEEPARstart{D}{riven} by the advancements in artificial intelligence (AI), wireless communications, and low-altitude intelligence, autonomous systems featuring unmanned aerial vehicles (UAVs) are becoming integral to a myriad of applications, including intelligent surveillance, disaster response, and logistics \cite{8660516,11006980}. These systems typically operate in a closed loop of sensing, communication, computation, and action. A critical bottleneck in such UAV-assisted networks is the transmission of high-dimensional sensor data, such as real-time video streams, from the UAV to a ground decision-making entity \cite{10275111}. This process is exceedingly demanding on both bandwidth and UAV's limited energy resources, posing a significant challenge to mission endurance and scalability.

To overcome this limitation, semantic communication has emerged as a transformative paradigm \cite{9955312,9921202,10915662}. Unlike traditional communication systems that focus on bit-level fidelity, semantic communication aims to extract and transmit only the essential, task-relevant meaning embedded within the source data. This approach can drastically reduce the volume of transmit data, thereby enhancing both spectral and energy efficiency. Concurrently, the rise of agentic AI, mostly powered by large AI models, provides the ``brain" for these autonomous systems, enabling sophisticated reasoning and decision-making based on the received semantic information \cite{zhang2025toward}. The integration of semantic communication with agentic AI thus paves the way for highly intelligent and efficient low-altitude wireless networks \cite{10972017}.

While the foundational concepts of semantic communication and AI-driven control have been explored, a holistic understanding of the system-level design and resource allocation remains a significant challenge. Prior works have focused on either the semantic codec design \cite{10768970} or AI task offloading \cite{8533343}, without jointly considering the intricate interplay between the level of semantic compression, physical deployment of the UAV, and strategic offloading of AI inference tasks. For instance, deeper semantic compression reduces transmission energy but increases on-board computation energy. Similarly, offloading the AI inference task from the edge to a powerful cloud server saves local resources but incurs additional communication latency and energy costs. A comprehensive framework that jointly optimizes these highly-coupled variables to enhance the overall system efficiency is therefore urgently needed.

In this letter, we address this gap by designing an energy-efficient agentic AI-assisted semantic wireless network. We consider a collaborative system where a UAV performs semantic sensing, a mobile base station (BS) acts as an edge intelligence hub, and a robot executes control commands. The main contributions of this work are summarized as follows:
\begin{itemize}
    \item We propose a framework for an agentic AI-assisted low-altitude semantic wireless network that intelligently orchestrates a sense-communicate-decide-act workflow.
    \item We formulate a system-wide energy consumption minimization problem that jointly optimizes the UAV's three-dimensional (3D) location, the semantic compression ratio, the transmit powers of the UAV and BS, and the binary decision of AI inference task offloading (edge vs. cloud), under latency and quality-of-service (QoS) constraints. To solve the formulated non-convex problem, we develop an efficient algorithm which can find the globally optimal solution.
    \item Simulation results validate the superiority of the proposed framework in significantly reducing the system's energy consumption compared to baseline approaches.
\end{itemize}

\section{System Model and Problem Formulation}\label{Sec:smpf}
\subsection{Network Model}
\begin{figure}[t]
    \centering
    \includegraphics[width=\linewidth]{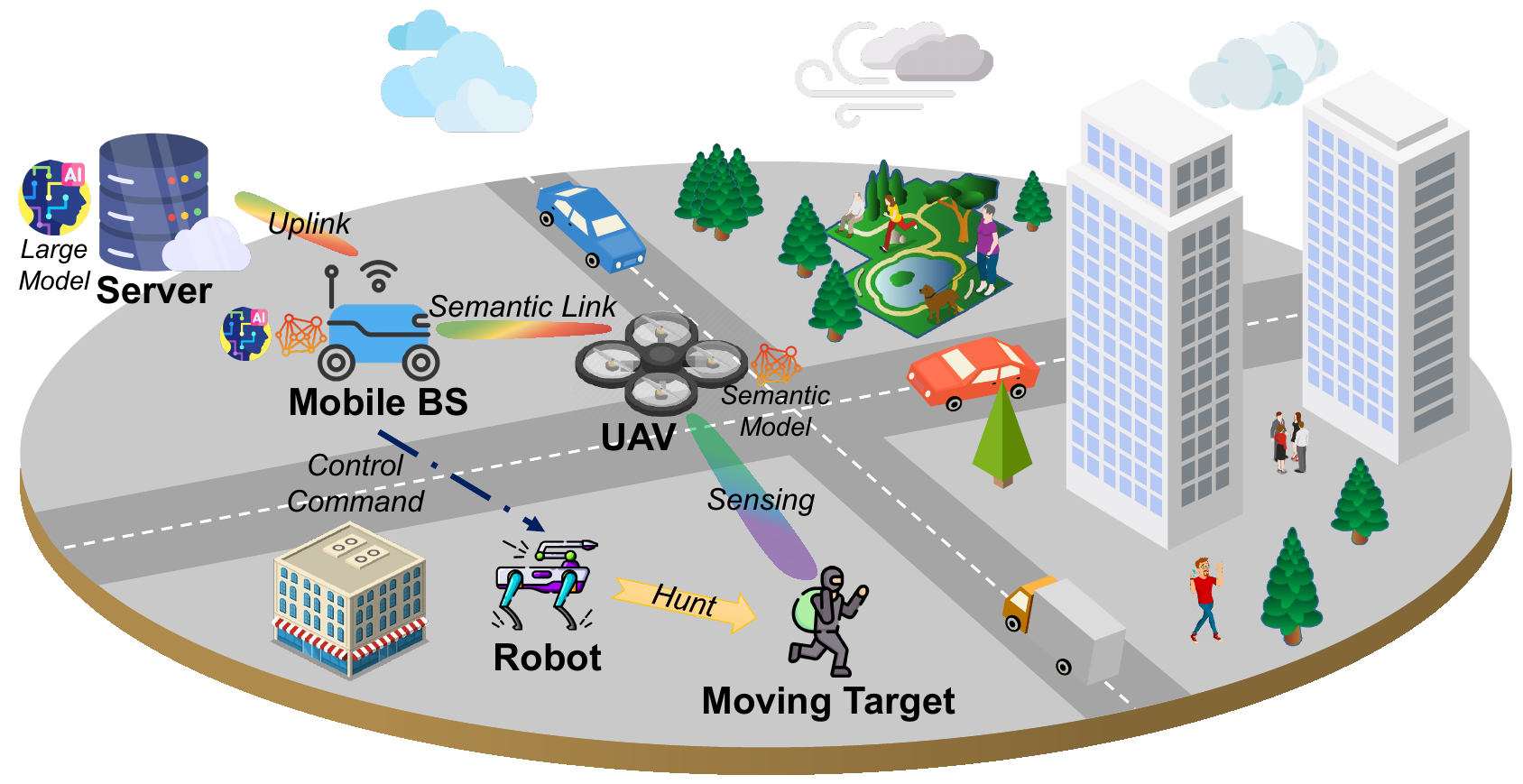}
    \caption{An illustration of the agentic AI-assisted semantic wireless network.}
    \label{fig.sm}
    \vspace{-1em}
\end{figure}

Consider an agentic AI-assisted semantic wireless network comprising a single UAV, a mobile BS, a cloud server, and a robot, as depicted in Fig.~\ref{fig.sm}. The system operates in a low-altitude environment where the UAV is tasked with tracking a moving target. The UAV senses information about the target, which is then semantically compressed and transmitted to the mobile BS. We consider a mobile BS to leverage its mobility, thereby maintaining proximity to the UAV and ensuring a robust communication link.
Upon receiving the semantic information, the mobile BS first reconstructs the original data using a paired semantic model. Subsequently, it can employ an embodied large AI model to make intelligent decisions and generate control commands based on the reconstructed information. Alternatively, the mobile BS can offload this task by transmitting the reconstructed data to a cloud server, which possesses superior computational resources, for command generation. The final control command is then transmitted to the robot to execute a specific task.

\begin{figure}[t]
    \centering
    \includegraphics[width=\linewidth]{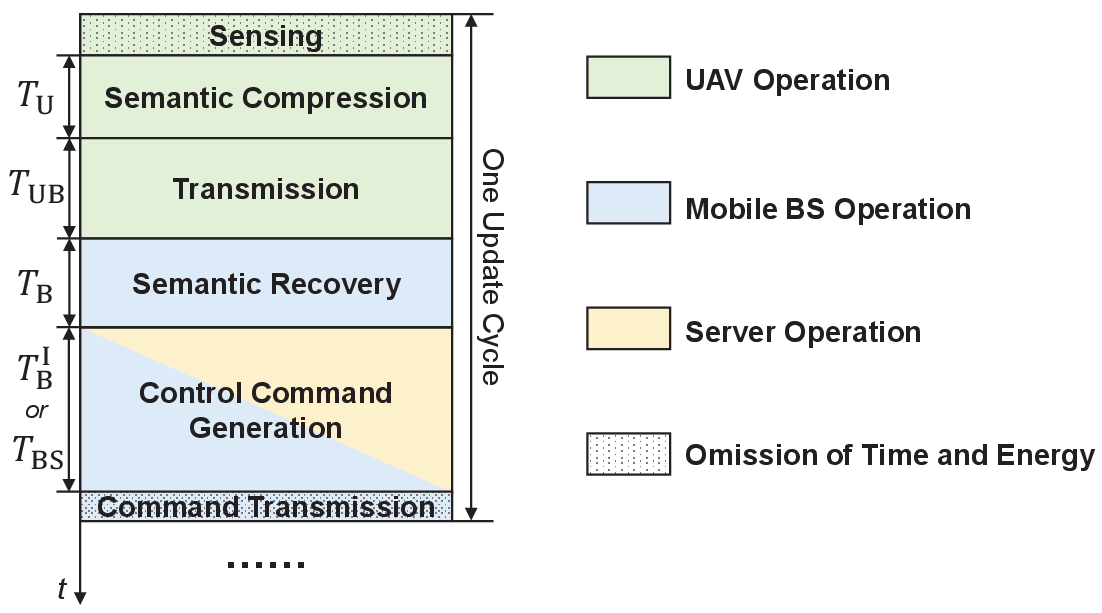}
    \caption{The operational workflow of the system within a single update cycle.}
    \label{fig.wp}
    \vspace{-1em}
\end{figure}

The operational workflow of the system is illustrated in Fig.~\ref{fig.wp}. We define an update cycle as the whole process of handling the data collected in one sensing interval. Each cycle involves several distinct stages: UAV sensing, on-board semantic compression, data transmission to the mobile BS, semantic recovery at the BS, control command generation, and command transmission to the robot. The control command generation can be performed either locally at the mobile BS or remotely at the cloud server.
Note that the time and energy for UAV sensing and the final command transmission to the robot are neglected, as they are assumed to be constant or negligible compared to the other components.

For analytical tractability, we establish a 3D Cartesian coordinate system where the mobile BS is located at the origin, i.e., $(0,0,0)$. The moving target's location in current time slot is denoted by $(x_\mathrm{T},y_\mathrm{T},0)$. The altitudes of both the mobile BS and the target are considered negligible relative to the UAV's operational altitude. The UAV's location is given by
\begin{equation}
    \mathbf{L}_\mathrm{U}=\left(x_\mathrm{U},y_\mathrm{U},H_\mathrm{U}\right),
\end{equation}
where $H_\mathrm{U}$ is the UAV's altitude.\footnote{The locations of all network components are assumed to be static within a single update cycle due to its short duration.} Consequently, the distance between the UAV and the mobile BS is $d_\mathrm{UB}=\sqrt{x_\mathrm{U}^2+y_\mathrm{U}^2+H_\mathrm{U}^2}$, and the distance between the UAV and the moving target is
\begin{equation}
    d_\mathrm{UT}=\sqrt{\left(x_\mathrm{U}-x_\mathrm{T}\right)^2+\left(y_\mathrm{U}-y_\mathrm{T}\right)^2+H_\mathrm{U}^2}.
\end{equation}

\subsection{UAV Sensing, Computation, and Communication Models}
Each update cycle commences with the UAV sensing the target, typically through video capture. The QoS for this sensing task is modeled as an exponentially decaying function of the sensing distance \cite{9130055}:
\begin{equation}
    q=e^{-\xi d_\mathrm{UT}},
\end{equation}
where $\xi$ is a parameter characterizing the sensing performance.

Let the raw data captured by the UAV be $\mathcal{D}$, with a total size of $D$ bits. This data undergoes semantic compression. We define the semantic compression ratio as $\rho=C/D$, where $C$ is the size of the compressed data. The computational overhead for compression with respect to $\rho$ is modeled as \cite{11096081}
\begin{equation}
    O(\rho)=-\kappa_1 D \ln(\rho),
\end{equation}
where $\kappa_1$ is the coefficient representing the computation efficiency of the semantic compression algorithm. The latency required for semantic compression is
\begin{equation}
    T_\mathrm{U}=\frac{-\kappa_1 D \ln(\rho)}{f_\mathrm{U}},
\end{equation}
where $f_\mathrm{U}$ is the computational capacity (in cycles/s) of the UAV's on-board processor. The corresponding energy consumption for this computation is
\begin{equation}
    E_\mathrm{U}=\tau_\mathrm{U} f_\mathrm{U}^3 T_\mathrm{U}, 
\end{equation}
where $\tau_\mathrm{U}$ is the effective switched capacitance coefficient of the UAV's processor.

After compression, the data is transmitted to the mobile BS. We assume the communication channel is dominated by the line-of-sight (LoS) path, which is a reasonable assumption for outdoor, low-altitude scenarios.\footnote{Specifically, we model a rural environment with a path-loss exponent of $\alpha=2$ and assume that signal power from antenna sidelobes is negligible.} The channel gain between the UAV and the BS is therefore
\begin{equation}
    G(\mathbf{L}_\mathrm{U})=\frac{G_0}{d_\mathrm{UB}^\alpha},
\end{equation}
where $G_0$ is the reference channel gain at a distance of $1$ m, and $\alpha$ is the path-loss exponent. Using the Shannon capacity formula, the transmission latency is
\begin{equation}
    T_\mathrm{UB}=\frac{D \rho}{B_\mathrm{U} \log_2\left(1+\frac{p_\mathrm{U} G(\mathbf{L}_\mathrm{U})}{B_\mathrm{U} N_0}\right)},
\end{equation}
where $B_\mathrm{U}$ is the channel bandwidth, $p_\mathrm{U}$ is the UAV's transmit power, and $N_0$ is the noise power spectral density. The energy consumed during this transmission is
\begin{equation}
    E_\mathrm{UB}=p_\mathrm{U} T_\mathrm{UB}.
\end{equation}

Upon reception, the mobile BS performs semantic recovery to reconstruct the original data. The time and energy for this process are modeled symmetrically to the compression phase:
\begin{equation}
    T_\mathrm{B}=\frac{-\kappa_2 D \ln(\rho)}{f_\mathrm{B}},
\end{equation}
where $\kappa_2$ and $f_\mathrm{B}$ are the semantic recovery efficiency parameter and the computational capacity of the mobile BS, respectively. The energy consumption for semantic recovery is
\begin{equation}
    E_\mathrm{B}=\tau_\mathrm{B} f_\mathrm{B}^3 T_\mathrm{B}. 
\end{equation}
where $\tau_\mathrm{B}$ is the effective switched capacitance coefficient of the mobile BS's processor.

\subsection{Control Command Generation Model}
Based on the recovered information, a control command for the robot is generated. This can occur either at the mobile BS or be offloaded to the cloud server.

\subsubsection{Case 1 (Generate at Mobile BS)}
The mobile BS uses its local large AI model for inference. We model the inference latency as being proportional to the original data size:
\begin{equation}
    T_\mathrm{B}^\mathrm{I}=\frac{\kappa_3 D}{f_\mathrm{B}},
\end{equation}
where $\kappa_3$ is the inference computation efficiency parameter. The energy consumed for task inference by the mobile BS is

\begin{equation}
    E_\mathrm{B}^\mathrm{I}=\tau_\mathrm{B} f_\mathrm{B}^3 T_\mathrm{B}^\mathrm{I}.
\end{equation}

\subsubsection{Case 2 (Generate at Cloud Server)}
The mobile BS transmits the recovered data to the cloud server. The latency required for this uplink transmission is
\begin{equation}
    T_\mathrm{BS}=\frac{D}{B_\mathrm{B} \log_2\left(1+\frac{p_\mathrm{B} G_\mathrm{BS}}{B_\mathrm{B} N_0}\right)},
\end{equation}
where $B_\mathrm{B}$ is the bandwidth for the BS-server link, $p_\mathrm{B}$ is the mobile BS's transmit power, and $G_\mathrm{BS}$ is the corresponding channel gain. The energy cost for this transmission is
\begin{equation}
    E_\mathrm{BS}=p_\mathrm{B} T_\mathrm{BS}.
\end{equation}
Given the substantial computational power of cloud servers, the inference latency at the server is considered negligible. Similarly, the server's energy consumption is not factored into our model as it is assumed to have a constant power supply.

To distinguish these two cases, we introduce a binary decision variable $a \in \{0, 1\}$. If generation is performed locally at the mobile BS, $a=1$; otherwise, $a=0$.

\subsection{Problem Formulation}
The total energy consumption of the system per update cycle is the sum of the energy consumed for UAV computation and transmission, BS computation, and the task-dependent energy for command generation, i.e.,
\begin{equation}
    E_\mathrm{tot}=E_\mathrm{U}+E_\mathrm{UB}+E_\mathrm{B}+a E_\mathrm{B}^\mathrm{I}+(1-a)E_\mathrm{BS}.
\end{equation}

Our objective is to minimize the total energy consumption by jointly optimizing the UAV's 3D location, the semantic compression ratio, the transmit powers, and the command generation offloading decision. Mathematically, this optimization problem is formulated as follows:
\begin{subequations}\label{eq.pf}
    \begin{align}
		\min_{a,\mathbf{L}_\mathrm{U},\rho,p_\mathrm{U},p_\mathrm{B}} \quad &  E_\mathrm{tot} \tag{\theequation}\\
		\mathrm{s.t.}\hspace{2.3em}
        & T_\mathrm{U}+T_\mathrm{UB}+T_\mathrm{B}+aT_\mathrm{B}^\mathrm{I}+(1-a)T_\mathrm{BS} \leq T_{\mathrm{th}},\label{cons.T}\\
        & q\geq q_{\mathrm{th}},\label{cons.q}\\
        & H_{\min}\leq H_\mathrm{U} \leq H_{\max},\label{cons.H}\\
        & \rho_{\mathrm{th}}\leq\rho\leq 1,\label{cons.rho}\\
        & 0<p_\mathrm{U}\leq p_\mathrm{U}^{\max},\label{cons.pu}\\
        & 0<p_\mathrm{B}\leq p_\mathrm{B}^{\max},\label{cons.pb}\\
        & a \in \{0,1\}.\label{cons.a}
    \end{align}
\end{subequations}
Constraint \eqref{cons.T} imposes a maximum allowable latency, $T_{\mathrm{th}}$, for the entire update cycle. Constraint \eqref{cons.q} ensures a minimum required sensing QoS, $q_{\mathrm{th}}$. The UAV's altitude is constrained by \eqref{cons.H} to lie within a regulated range $[H_{\min}, H_{\max}]$. The semantic compression ratio is bounded by \eqref{cons.rho}, where $\rho_{\mathrm{th}}$ is a lower limit to maintain semantic integrity. Constraints \eqref{cons.pu} and \eqref{cons.pb} define the maximum transmit powers for the UAV and mobile BS, respectively. Finally, \eqref{cons.a} specifies the binary nature of the task offloading variable.

The mixed-integer non-convex problem \eqref{eq.pf} is NP-hard. In the next section, we exploit the unique structure of problem \eqref{eq.pf} to find its globally optimal solution via an efficient two-dimensional search.

\section{Algorithm Design}\label{Sec:ad}

\subsection{Optimal UAV Location}
Intuitively, to minimize both energy consumption and latency, the UAV should be positioned as close as possible to the mobile BS, subject to the operational constraints. A shorter distance $d_\mathrm{UB}$ enhances the channel gain, which in turn reduces the required transmission power and time. This insight allows us to decouple the UAV placement problem. The optimal UAV location, $\mathbf{L}_\mathrm{U}^*$, can be found by solving the following convex optimization problem:
\begin{subequations}\label{eq.oul}
    \begin{align}
		\min_{\mathbf{L}_\mathrm{U}} \quad &  d_\mathrm{UB} \tag{\theequation}\\
		\mathrm{s.t.}\hspace{1em}
        & d_\mathrm{UT}\leq \frac{-\ln(q_{\mathrm{th}})}{\xi},\label{cons.q1}\\
        & H_{\min}\leq H_\mathrm{U} \leq H_{\max}.\label{cons.H1}
    \end{align}
\end{subequations}
Problem \eqref{eq.oul} seeks to minimize the Euclidean distance to the origin subject to a convex set defined by the intersection of a sphere (representing the QoS constraint) and a slab (representing the altitude constraint). Let $D_{\max} = \frac{-\ln(q_{\mathrm{th}})}{\xi}$ be the maximum permissible distance from the UAV to the moving target. Assuming feasibility, i.e., $H_{\min} \leq D_{\max}$, the solution can be derived geometrically.

\begin{theorem}\label{theo1}
The optimal UAV altitude is the minimum allowable altitude:
\begin{equation}
    H_\mathrm{U}^* = H_{\min}.
\end{equation}
The optimal horizontal coordinates $(x_\mathrm{U}^*, y_\mathrm{U}^*)$ are determined by one of two cases:

\subsubsection{Case 1 ($x_\mathrm{T}^2 + y_\mathrm{T}^2 \le D_{\max}^2 - H_{\min}^2$)}
In this case, the target is sufficiently close to the BS. The optimal UAV location is directly above the BS at the minimum altitude:
        \begin{equation}
            \mathbf{L}_\mathrm{U}^* = (0, 0, H_{\min}).
        \end{equation}

\subsubsection{Case 2 ($x_\mathrm{T}^2 + y_\mathrm{T}^2 > D_{\max}^2 - H_{\min}^2$)}
In this case, the UAV must move towards the target to satisfy the QoS constraint. The optimal location lies on the edge of the cylindrical service area in the direction of the target:
        \begin{align}
            x_\mathrm{U}^* &= x_\mathrm{T} \left(1 - \frac{\sqrt{D_{\max}^2 - H_{\min}^2}}{\sqrt{x_\mathrm{T}^2 + y_\mathrm{T}^2}}\right), \\
            y_\mathrm{U}^* &= y_\mathrm{T} \left(1 - \frac{\sqrt{D_{\max}^2 - H_{\min}^2}}{\sqrt{x_\mathrm{T}^2 + y_\mathrm{T}^2}}\right).
        \end{align}
\end{theorem}

\subsection{Optimal Semantic Compression Ratio and Transmit Power}

\subsubsection{Optimal Semantic Compression Ratio with Given Transmit Power}
With the optimal UAV location $\mathbf{L}_\mathrm{U}^*$, a fixed task offloading indicator $a$ and fixed transmit powers $p_\mathrm{U}$ and $p_\mathrm{B}$, the subproblem for the semantic compression ratio $\rho$ is
\begin{subequations}\label{eq.rho}
    \begin{align}
		\min_{\rho} \quad &  A\rho - F \ln(\rho) \tag{\theequation}\\
		\mathrm{s.t.}\hspace{1em}
        & K_1 \rho + K_2 \ln(\rho) \le T,\label{cons.T2}\\
        & \rho_{\mathrm{th}}\leq\rho\leq 1,\label{cons.rho2}
    \end{align}
\end{subequations}
where
\begin{align*}
    A & \triangleq \frac{p_\mathrm{U} D}{B_\mathrm{U} \log_2\left(1+\frac{p_\mathrm{U} G(\mathbf{L}_\mathrm{U}^*)}{B_\mathrm{U} N_0}\right)}, \  F \triangleq D (\tau_\mathrm{U} \kappa_1 f_\mathrm{U}^2 + \tau_\mathrm{B} \kappa_2 f_\mathrm{B}^2), \\
    K_1 & \triangleq \frac{D}{B_\mathrm{U} \log_2\left(1+\frac{p_\mathrm{U} G(\mathbf{L}_\mathrm{U}^*)}{B_\mathrm{U} N_0}\right)}, \  K_2 \triangleq -D\left(\frac{\kappa_1}{f_\mathrm{U}} + \frac{\kappa_2}{f_\mathrm{B}}\right), \\
    T & \triangleq T_{\mathrm{th}} - \left( \frac{a \kappa_3 D}{f_\mathrm{B}} + \frac{(1-a)D}{B_\mathrm{B} \log_2\left(1+\frac{p_\mathrm{B} G_\mathrm{BS}}{B_\mathrm{B} N_0}\right)} \right).
\end{align*}

\begin{theorem}\label{theo2}
The optimal semantic compression ratio of problem \eqref{eq.rho} is  
\begin{equation}\label{rhostar}
     \rho^*\left(p_\mathrm{U},p_\mathrm{B}\right) = \min\left\{\max\left\{\rho_0, \rho_{\mathrm{th}}, \rho_a\right\}, 1, \rho_b \right\},
\end{equation}
where $\rho_0 = \frac{F}{A}$, $\rho_a < \rho_b$ are the roots of $K_1 \rho + K_2 \ln(\rho) = T$.
\end{theorem}

\begin{proof}
\emph{
Problem \eqref{eq.rho} admits a closed-form solution. The unconstrained minimizer of the objective is $\rho_0$. The latency constraint \eqref{cons.T2} defines a feasible interval $\left[\rho_a, \rho_b\right]$. These roots can be expressed using the principal ($W_0$) and secondary ($W_{-1}$) branches of the Lambert W function:
\begin{align}
    \rho_a = \frac{K_2}{K_1} W_0\left(\frac{K_1}{K_2} e^{\frac{T}{K_2}}\right),
    \rho_b = \frac{K_2}{K_1} W_{-1}\left(\frac{K_1}{K_2} e^{\frac{T}{K_2}}\right).
\end{align}
The optimal $\rho^*$ is then found by projecting the unconstrained minimizer $\rho_0$ onto the final feasible region defined by the intersection of $\left[\rho_a, \rho_b\right]$ and $\left[\rho_{\mathrm{th}}, 1\right]$, leading to Eq. \eqref{rhostar}.
} \hfill$\blacksquare$
\end{proof}

Crucially, Theorem \ref{theo2} establishes the optimal semantic compression ratio as an explicit function of the transmit powers.

\subsubsection{Two-Dimensional Search on Transmit Powers}
By substituting the analytical solutions for the optimal UAV location from Theorem \ref{theo1} and the optimal semantic compression ratio from Theorem \ref{theo2} back into problem \eqref{eq.pf}, the original multi-variable optimization problem is reduced to a two-dimensional problem over the transmit powers for a fixed $a$:
\begin{subequations}\label{eq.pp}
    \begin{align}
		\min_{p_\mathrm{U},p_\mathrm{B}} &  \frac{p_\mathrm{U} D \rho^*\left(p_\mathrm{U},p_\mathrm{B}\right)}{B_\mathrm{U} \log\left(1+\frac{p_\mathrm{U} G(\mathbf{L}_\mathrm{U}^*)}{B_\mathrm{U} N_0}\right)} + \frac{p_\mathrm{B} D (1-a)}{B_\mathrm{B} \log\left(1+\frac{p_\mathrm{B} G_\mathrm{BS}}{B_\mathrm{B} N_0}\right)}\tag{\theequation}\\
		\mathrm{s.t.}\hspace{0.2em}
        & T_\mathrm{U}+T_\mathrm{UB}+T_\mathrm{B}+aT_\mathrm{B}^\mathrm{I}+(1-a)T_\mathrm{BS} \leq T_{\mathrm{th}},\label{cons.Tpp}\\
        & 0<p_\mathrm{U}\leq p_\mathrm{U}^{\max},\label{cons.pupp}\\
        & 0<p_\mathrm{B}\leq p_\mathrm{B}^{\max}.\label{cons.pbpp}
    \end{align}
\end{subequations}
Due to the intricate form of $\rho^*\left(p_\mathrm{U}, p_\mathrm{B}\right)$, which involves the Lambert W function, the objective and constraint functions in problem \eqref{eq.pp} are highly non-convex, and an analytical solution for the optimal powers is intractable. Therefore, we find the solution by performing a two-dimensional search over the feasible power ranges. This method guarantees global optimality for the reduced problem and remains computationally feasible.

\subsection{Optimal Task Offloading Decision}
The task offloading variable $a$ is binary. Given that our previous steps find the optimal solution for a fixed $a$, we can determine the optimal $a^*$ by solving the problem for each case and selecting the one that yields a lower total energy consumption. Let $V_0$ and $V_1$ be the minimum energy values obtained from the two-dimensional search for $a=0$ and $a=1$, respectively. The optimal offloading decision is:
\begin{equation}
    a^*=
    \begin{cases}
        0,&\text{if }V_0 < V_1,\\
        1,&\text{otherwise}.
    \end{cases}
\end{equation}

\subsection{Overall Algorithm Procedure}
The overall algorithm to solve problem \eqref{eq.pf} with global optimality is outlined in Algorithm \ref{algo}.

\begin{algorithm}[t]
\caption{Solving Problem \eqref{eq.pf} with Optimal Solution}
\begin{algorithmic}[1]\label{algo}
\STATE Obtain the optimal UAV location $\mathbf{L}_\mathrm{U}^*$ using Theorem \ref{theo1}.
\STATE Obtain the optimal semantic compression ratio $\rho^*\left(p_\mathrm{U},p_\mathrm{B}\right)$ using Theorem \ref{theo2}.
\FOR{$a \in \{0, 1\}$}
    \STATE Solve problem \eqref{eq.pp} via two-dimensional search.
    \STATE Store the minimum total energy as $V_a$.
\ENDFOR
\STATE Compare $V_0$ and $V_1$ to determine $a^*$ and the corresponding optimal variables.
\end{algorithmic}
\end{algorithm}

\section{Simulation Results}\label{Sec:sr}
In the simulations, the moving target is located at (\SI{300}{\metre}, \SI{100}{\metre}, \SI{0}{\metre}), and the UAV's operational altitude is constrained to [\SI{40}{\metre}, \SI{400}{\metre}]. The maximum transmit power for both the UAV and the mobile BS is set to \SI{30}{dBm}. Unless specified otherwise, the default parameters include a BS-server link bandwidth of \SI{0.5}{MHz}, a raw data size of \SI{1}{Mb}, and a maximum allowable latency of \SI{700}{ms}.

We evaluate the performance of our proposed algorithm against several benchmark schemes: \textbf{Generate at the Server/BS}, which consistently offloads the task to the server or processes it at the BS, respectively; \textbf{Non-Semantic}, which transmits the raw data without any semantic compression; \textbf{Max Transmit Power}, which operates the UAV and BS at their maximum power levels; \textbf{Fixed UAV Location}, which places the UAV at $(x_\mathrm{T}, y_\mathrm{T}, H_{\min})$; and \textbf{BCD Algorithm}, which employs block coordinate descent (BCD) to alternately optimize the semantic compression ratio and transmit powers.

Fig.~\ref{fig.sim} illustrates the total energy consumption as a function of the BS-server link bandwidth, raw data size, and maximum allowable latency. Across all simulated scenarios, the proposed algorithm consistently achieves the lowest energy consumption, which validates its superior performance in joint resource allocation. As depicted in the figures, the `Max Transmit Power' scheme is uniformly inefficient, underscoring the necessity of adaptive power control. The `Fixed UAV Location' scheme exhibits a consistent performance gap, which highlights the energy savings gained from optimizing the UAV's deployment. Furthermore, the `BCD Algorithm' occasionally converges to a local optimum, demonstrating the advantage of our proposed method in attaining a globally optimal solution.

More specifically, Fig.~\ref{fig.sim}(a) shows that as the BS-server bandwidth increases, the energy cost of offloading to the server decreases, making it the preferable strategy. Our algorithm adaptively selects this optimal offloading decision, outperforming the fixed `Generate at the Server/BS' scheme. As seen in Fig.~\ref{fig.sim}(b) and Fig.~\ref{fig.sim}(c), the `Non-Semantic' approach performs poorly with large data sizes or stringent latency constraints, due to the excessive energy required for transmitting uncompressed data. This confirms the significant efficiency gains provided by semantic communications.

\begin{figure*}[t]
    \centering
    \vspace{-2em}
    \subfigure{
        \begin{minipage}{0.33\textwidth}
            \centering
            \includegraphics[width=1\textwidth]{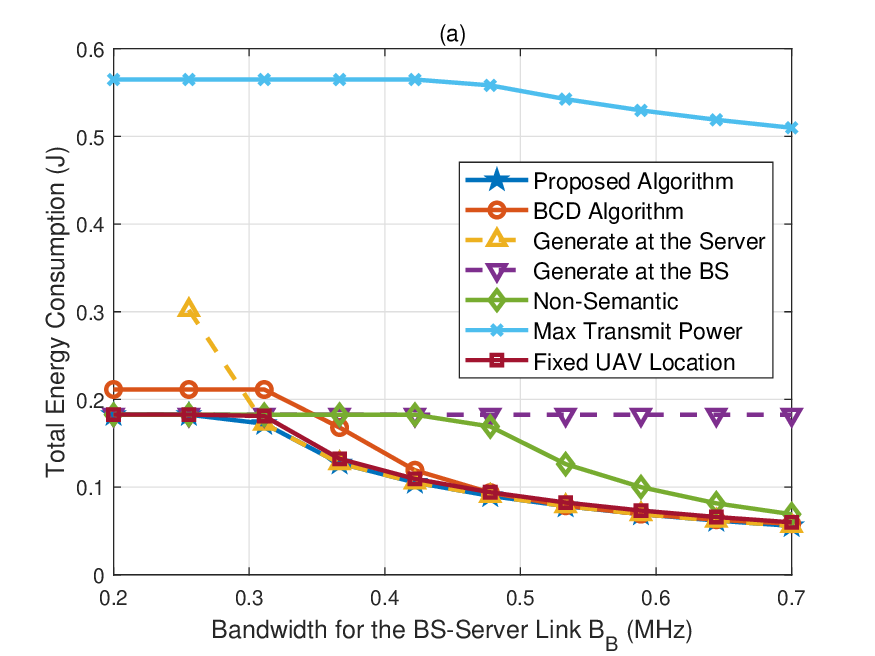}
            \label{fig.BB}
    \end{minipage}}
    \hspace{-5mm}
    \subfigure{
        \begin{minipage}{0.33\textwidth}
            \centering
            \includegraphics[width=1\textwidth]{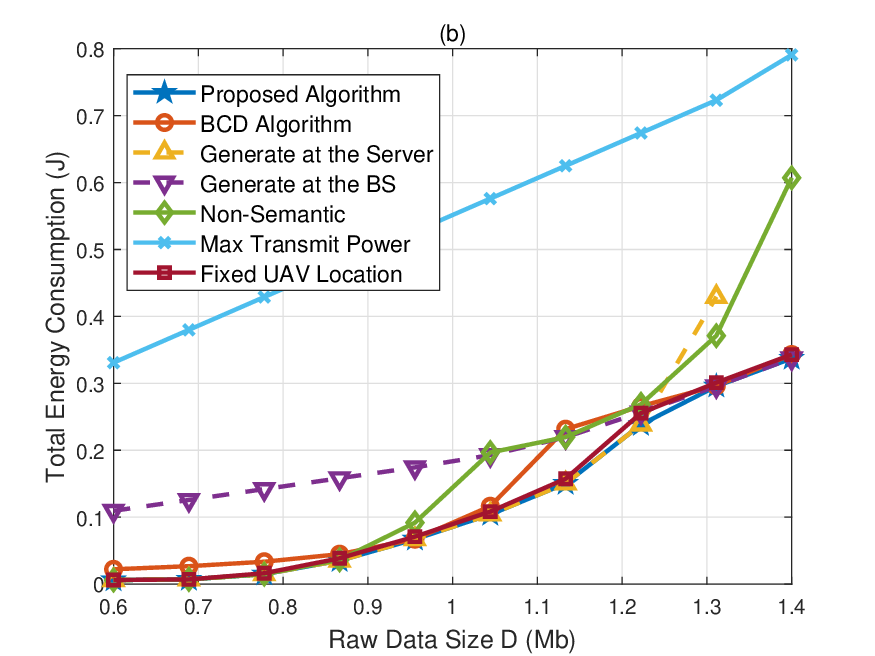}
            \label{fig.D}	
    \end{minipage}}
        \hspace{-5mm}
        \subfigure{
            \begin{minipage}{0.33\textwidth}
                \centering
                \includegraphics[width=1\textwidth]{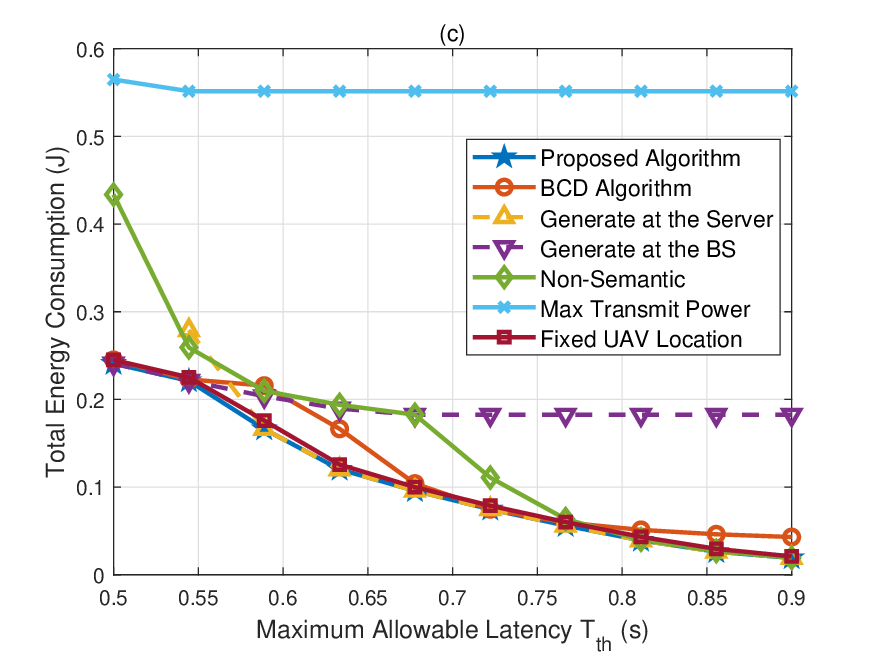}
                \label{fig.Tth}	
        \end{minipage}}
    \vspace{-1.4em}
    \caption{Total energy consumption versus: (a) Bandwidth for the BS-server link $B_\mathrm{B}$, (b) Raw data size $D$, (c) Maximum allowable latency $T_{\mathrm{th}}$.}
    \label{fig.sim} 
    \vspace{-1em}
\end{figure*}

\section{Conclusion}\label{Sec:c}
In this letter, we designed an energy-efficient framework for an agentic AI-assisted low-altitude semantic wireless network. We formulated a comprehensive system-wide energy minimization problem that captures the intricate trade-offs between communication, computation, and sensing performance. By jointly optimizing the UAV's 3D placement, the level of semantic compression, transmit power control, and the AI inference task offloading, our approach aims to enhance the operational endurance of the system. We developed an efficient algorithm to obtain the globally optimal solution for this non-convex optimization problem. 
Future research could extend this framework to multi-UAV, multi-target scenarios, incorporate dynamic mobility models, and explore reinforcement learning-based methods for real-time resource allocation in more complex and uncertain environments.

\bibliographystyle{IEEEtran}
\bibliography{ref}

\end{document}